\documentclass[a4paper]{amsart}

\usepackage{amsmath}
\usepackage{amssymb}
\usepackage{amsthm}

\usepackage{thmtools} %Restatables

\usepackage{ stmaryrd }
\usepackage{mdwtab}
\usepackage{multirow}
\usepackage{mathabx}
\usepackage{synttree}
\usepackage{proof}
\usepackage{graphicx}
\usepackage{tikz-cd}
\usepackage{slashed}
\usepackage{footnote}
\usepackage{pdfpages}

\usepackage{xfrac}

\usepackage{natbib}

\makeatletter
\def\blx@maxline{77}
\makeatother

\usepackage[utf8]{inputenc}

\usepackage[colorlinks,citecolor=brown,linkcolor=blue]{hyperref}
\definecolor{DarkGreen}{RGB}{0,100,0}
\definecolor{DarkBlue}{RGB}{0,0,200}
\newcommand{\ei}[1]{{\textcolor{cyan}{[{#1}---Egor]}}}
\newcommand{\lo}[1]{{\textcolor{red}{[{#1}---Luke]}}}

\newif\ifdraft\drafttrue
\usepackage{microtype}

\usepackage{geometry}
\usepackage{enumitem}

 \usepackage[capitalise]{cleveref}

\newtheorem{theorem}{Theorem}[section]
\newtheorem{proposition}[theorem]{Proposition}

\theoremstyle{definition}
\newtheorem{definition}[theorem]{Definition}

\renewcommand{\setminus}{\backslash}

\newcommand{\ria}{\rightarrow}

\newcommand{\CP}{\mathcal{P}}

\newbox\gnBoxA
\newdimen\gnCornerHgt
\setbox\gnBoxA=\hbox{$\ulcorner$}
\global\gnCornerHgt=\ht\gnBoxA
\newdimen\gnArgHgt
\def\Godelnum #1{%
\setbox\gnBoxA=\hbox{$#1$}%
\gnArgHgt=\ht\gnBoxA%
\ifnum     \gnArgHgt<\gnCornerHgt \gnArgHgt=0pt%
\else \advance \gnArgHgt by -\gnCornerHgt%
\fi \raise\gnArgHgt\hbox{$\ulcorner$} \box\gnBoxA %
\raise\gnArgHgt\hbox{$\urcorner$}}



\newcommand\set[1]{{ \{\, #1 \,\} }}

\title{Two statements of the Duggan-Schwartz theorem}
\author{Egor Ianovski}

\newcommand{\best}{\textnormal{best}}
\newcommand{\worst}{\textnormal{worst}}

\begin{document}

\maketitle

\begin{abstract}
The Duggan-Schwartz theorem \citep{Duggan1992} is a famous result concerning strategy-proof
social choice correspondences, often stated as ``A social choice correspondence that can be manipulated
by neither an optimist nor a pessimist has a weak dictator''. However, this formulation is actually due to \citet{Taylor2002},
and the original theorem, at face value, looks rather different. In this note we show that the two are in
fact equivalent.
\end{abstract}

\section{Definitions}

\begin{definition}
    Let $V$ be a finite set of voters, $A$ a finite set of alternatives.
    
    A profile $P$ consists of a linear order over $A$ (also known
    as a \emph{preference order} or a \emph{ballot}), $P_i$, for every voter $i$. The set of all profiles of voters $V$ over alternatives $A$ is denoted $\CP(V,A)$. We use $P_{-i}$ to refer to the ballots of all voters except $i$. Hence,
    $P=P_iP_{-i}$ and $P_i'P_{-i}$ is obtained from profile $P$ by replacing
    $P_i$ with $P_i'$.

A \emph{social choice correspondence} produces a nonempty set
    of alternatives, $F:\CP(V,A)\ria 2^A\setminus\set{\emptyset}$.
\end{definition}

\begin{definition}\label{def:sp}
    Let $\emptyset\neq W\subseteq A$. We use $\text{best}(P_i,W)$ to denote the best alternative in $W$ according to $P_i$, $\text{worst}(P_i,W)$ the worst.
    
    We extend $\succeq_i$ into two weak orders over $2^A\setminus\set{\emptyset}$:
    \begin{enumerate}
        \item 
$X\succeq_i^O Y$ iff $\best(P_i,X)\succeq_i\best(P_i,Y)$.
        \item 
$X\succeq_i^P Y$ iff $\worst(P_i,X)\succeq_i\worst(P_i,Y)$.
    \end{enumerate}
    
    A social choice correspondence is \emph{strategy-proof for optimists} (SPO) if for all $P_i'$, whenever $F(P_iP_{-i})=W$ and $F(P_i'P_{-i})=W'$,
    $W\succeq_i^O W'$.   
    A social choice correspondence is \emph{strategy-proof for pessimists} (SPP) if for all $P_i'$, whenever $F(P_iP_{-i})=W$ and $F(P_i'P_{-i})=W'$,
    $W\succeq_i^P W'$.
\end{definition}

\begin{definition}
    Given a social choice correspondence $F$, a \emph{weak dictator} is some $i\in V$ such that the first choice of $i$ is always in $F(P)$.    
\end{definition}

\section{Proofs}

\begin{theorem}[\cite{Taylor2002}]\label{thm:Taylor}
    Let $F$ be a social choice correspondence that satisfies SPP, SPO and is onto with respect
    to singletons. That is, for every $a\in A$ there exists a $P$ such that $F(P)=\set{a}$.
    
    For $|A|\geq 3$, $F$ has a weak dictator.
\end{theorem}

\begin{theorem}[\cite{Duggan1992}]
    Let $F$ be a social choice correspondence that is onto with respect
    to singletons. That is, for every $a\in A$ there exists a $P$ such that $F(P)=\set{a}$.
    
    Let each voter $i$ be equipped with a probability
    function $p_i:\CP(V,A)\times A\times 2^A\ria [0,1]$ such that $\sum_{x\in X}p_i(P,x,X)=1$ and $p_i(P,a,X)>0$ whenever $a=\best(P_i,X)$
    or $a=\worst(P_i,X)$.
    
    Suppose further that for every $u_i$ consistent with $P_i$ 
    ($u_i(a)>u_i(b)$ whenever $a\succ_i b$), and for every $P_i'$, the following is true:
    $$\sum_{x\in F(P_iP_{-i})}p_i(P_iP_{-i},x,F(P_iP_{-i}))u_i(x)\geq
    \sum_{x\in F(P_i'P_{-i})}p_i(P_iP_{-i},x,F(P_i'P_{-i}))u_i(x).$$
    
    For $|A|\geq 3$, $F$ has a weak dictator.
\end{theorem}

The notion of manipulation used by \citeauthor{Duggan1992} is obviously
more general than that of \citeauthor{Taylor2002}, and one is thus
tempted to conclude that the original theorem is weaker than
Taylor's reformulation.\footnote{A wider notion of manipulability implies
a more narrow notion of strategy-proofness, and hence the theorem
would apply to less functions.}

However, this would be erroneous as the theorems, strictly speaking,
are incomparable. Taylor's theorem concerns a social choice correspondence
$F$, whereas \citeauthor{Duggan1992}'s theorem applies to $F$ \emph{together}
with a set of probability functions, $p_i$. It is entirely plausible
that one could find two sets of probability functions such that $F$
and $p_1,\dots,p_n$ satisfy the hypotheses of the Duggan-Schwartz
theorem while $F$ and $p_1',\dots,p_n'$ do not. However, $F$
is unchanged -- it either has a weak dictator, or it does not.

To more properly compare the two theorems, then, we need to take
an existential projection over the original Duggan-Schwartz
theorem.

\begin{theorem}[\cite{Duggan1992}]\label{thm:DSexist}
    Let $F$ be a social choice correspondence that is onto with respect
    to singletons. That is, for every $a\in A$ there exists a $P$ such that $F(P)=\set{a}$.
    
    Suppose there exist probability
    functions $p_i:\CP(V,A)\times A\times 2^A\ria [0,1]$ such that $\sum_{x\in X}p_i(P,x,X)=1$ and $p_i(P,a,X)>0$ whenever $a=\best(P_i,X)$
    or $a=\worst(P_i,X)$.
    
    Suppose further that for every $u_i$ consistent with $P_i$ 
    ($u_i(a)>u_i(b)$ whenever $a\succ_i b$), and for every $P_i'$, the following is true:
    $$\sum_{x\in F(P_iP_{-i})}p_i(P_iP_{-i},x,F(P_iP_{-i}))u_i(x)\geq
    \sum_{x\in F(P_i'P_{-i})}p_i(P_iP_{-i},x,F(P_i'P_{-i}))u_i(x).$$
    
    For $|A|\geq 3$, $F$ has a weak dictator.
\end{theorem}

Now we claim the two theorems are equivalent.

\begin{proposition}
    $F$ satisfies the hypotheses of \cref{thm:Taylor} if and only if
    $F$ satisfies the hypotheses of \cref{thm:DSexist}.
\end{proposition}
\begin{proof}
    We will first show that if $F$ is manipulable in the sense
    of Taylor it is manipulable in the sense of Duggan-Schwartz.
    Pay heed to the order of the quantifiers in \cref{thm:DSexist},
    as they may appear counter-intuitive:
    $F$ is strategy-proof if for \emph{some} choice of
    probability functions, for \emph{every} choice of a
    utility function, voter $i$ cannot improve his expected
    utility. Hence, $F$ is manipulable just if for \emph{every}
    choice of probability functions we can construct \emph{some}
    utility function giving voter $i$ a profitable deviation.
    
    Suppose $i$ can manipulate optimistically from $P_iP_{-i}$
    to $P_i'P_{-i}$. That is:
    \begin{align*}
    F(P_iP_{-i})=X&,\quad F(P_i'P_{-i})=Y,\\
    \best(P_i,X)=a&,\quad\best(P_i,Y)=b,\\
    a&\prec_i b.
    \end{align*}
    
    Now, let $p_i$ be any probability function in the sense
    of \cref{thm:DSexist}. Note that this means that $p_i(P,a,X)=\epsilon$
    and $p_i(P,b,Y)=\delta$ are strictly positive. Let $c\in X$ be the next-best alternative after $a$. Observe that an upper bound
    on the utility voter $i$ obtains sincerely is $\epsilon u_i(a)+(1-\epsilon)u_i(c)$, whereas the lower bound on
    the utility voter $i$ obtains from the deviation is $\delta u_i(b)$.
    All we need to do is pick a $u_i$ that satisfies:
    $$\delta u_i(b)>\epsilon u_i(a)+(1-\epsilon)u_i(c).$$
    It is of course easy to do so as, necessarily, $u_i(b)>u_i(a),u_i(c)$,
    and $\epsilon,\delta$ are constants. For example, let $u_i(a)=1,u_i(c)=2$ and $u_i(b)=\sfrac{3}{\delta}$.
    
    Suppose $i$ can manipulate pessimistically from $P_iP_{-i}$
    to $P_i'P_{-i}$. That is:
    \begin{align*}
    F(P_iP_{-i})=X&,\quad F(P_i'P_{-i})=Y,\\
    \worst(P_i,X)=a&,\quad\worst(P_i,Y)=b,\\
    a&\prec_i b.
    \end{align*}
    
    As before, let $p_i$ be any probability function in the sense
    of \cref{thm:DSexist}. This means that $p_i(P,a,X)=\epsilon$
    and $p_i(P,b,Y)=\delta$ are strictly positive. Let $c\in X$ be the best alternative in the set. Observe that an upper bound
    on the utility voter $i$ obtains sincerely is $\epsilon u_i(a)+(1-\epsilon)u_i(c)$, whereas the lower bound on
    the utility voter $i$ obtains from the deviation is $u_i(b)$.\footnote{$b$ is the worst element in $Y$, so the utility of
    any other element must be at least $u_i(b)$, and $p_i(P,y,Y)$ sums to 1.}
    All we need to do is pick a $u_i$ that satisfies:
    $$u_i(b)>\epsilon u_i(a)+(1-\epsilon)u_i(c).$$
    This time it is possible that $u_i(c)>u_i(b)$, however
    $1-\epsilon$ is strictly smaller than 1. One possibility
    is $u_i(a)=1$, 
    $u_i(b)=\frac{\sfrac{1}{\epsilon}+\epsilon+1}{1-\epsilon}$, $u_i(c)=\frac{\sfrac{1}{\epsilon}+\epsilon+2}{1-\epsilon}$.
    This leads to the following inequality, which can be verified
    algebraically:
    $$\frac{\sfrac{1}{\epsilon}+\epsilon+1}{1-\epsilon}>2\epsilon+2+\sfrac{1}{\epsilon}.$$
    
    Now suppose that $F$ is manipulable in the sense of Duggan-Schwartz.
    This means for every choice of $p_i$, there is some choice of $u_i$
    such that for some choice of $P_iP_{-i}$ and $P_i'P_{-i}$, $i$'s expected
    utility is higher in the insincere profile.
    
    Pick a $p_i$ that attaches a probability of $\sfrac{1}{2}$ to the
    best alternative in the set and $\sfrac{1}{2}$ to the worst.
    In other words, we have the following situation:
    \begin{align*}
    F(P_iP_{-i})=X&,\quad F(P_i'P_{-i})=Y,\\
    \best(P_i,X)=x_1&,\quad\best(P_i,Y)=y_1,\\
    \worst(P_i,X)=x_2&,\quad\worst(P_i,Y)=y_2,\\
    \frac{u_i(x_1)+u_i(x_2)}{2}&<\frac{u_i(y_1)+u_i(y_2)}{2}.\\
    \end{align*}
    Clearly, a necessary condition for the above to hold is that
    either $u_i(y_1)>u_i(x_1)$ or $u_i(y_2)>u_i(x_2)$. That is
    to say, $F$ is manipulable by either an optimist or a pessimist.
\end{proof}

\bibliographystyle{plainnat}             
\bibliography{references}  

\begin{thebibliography}{2}
\providecommand{\natexlab}[1]{#1}
\providecommand{\url}[1]{\texttt{#1}}
\expandafter\ifx\csname urlstyle\endcsname\relax
  \providecommand{\doi}[1]{doi: #1}\else
  \providecommand{\doi}{doi: \begingroup \urlstyle{rm}\Url}\fi

\bibitem[Duggan and Schwartz(1992)]{Duggan1992}
John Duggan and Thomas Schwartz.
\newblock Strategic manipulability is inescapable: Gibbard-satterthwaite
  without resoluteness.
\newblock Working Papers 817, California Institute of Technology, Division of
  the Humanities and Social Sciences, 1992.
\newblock URL \url{http://EconPapers.repec.org/RePEc:clt:sswopa:817}.

\bibitem[Taylor(2002)]{Taylor2002}
Alan~D. Taylor.
\newblock The manipulability of voting systems.
\newblock \emph{The American Mathematical Monthly}, 109\penalty0 (4):\penalty0
  321--337, 2002.
\newblock ISSN 00029890, 19300972.
\newblock URL \url{http://www.jstor.org/stable/2695497}.

\end{thebibliography}

\end{document}